\documentclass[10pt,twocolumn,twoside]{IEEEtran}

\usepackage{amsfonts}

\usepackage{epsfig}

\newtheorem{proposition}{\textit{Proposition}}
\newtheorem{algorithm}{\textit{Algorithm}}
\begin{document}
\title{An Efficient Greedy Algorithm for Sparse Recovery in Noisy Environment}
\author{Hao~Zhang, Gang~Li, Huadong~Meng
\thanks{Submitted August 13, 2009; }
\thanks{Hao Zhang, Gang Li and Huadong Meng are with the Department of Electronics Engineering, Tsinghua University.}}
%

\maketitle

\begin{abstract}
Greedy algorithm are in widespread use for sparse recovery because
of its efficiency. But some evident flaws exists in most popular
greedy algorithms, such as CoSaMP, which includes unreasonable
demands on prior knowledge of target signal and excessive
sensitivity to random noise. A new greedy algorithm called AMOP is
proposed in this paper to overcome these obstacles. Unlike CoSaMP,
AMOP can extract necessary information of target signal from sample
data adaptively and operate normally with little prior knowledge.
The recovery error of AMOP is well controlled when random noise is
presented and fades away along with increase of SNR. Moreover, AMOP
has good robustness on detailed setting of target signal and less
dependence on structure of measurement matrix. The validity of AMOP
is verified by theoretical derivation. Extensive simulation
experiment is performed to illustrate the advantages of AMOP over
CoSaMP in many respects. AMOP is a good candidate of practical
greedy algorithm in various applications of Compressed Sensing.
\end{abstract}


\section{Introduction}
\PARstart{S}parse signal recovery problem is the reconstruction of
such signals with characteristic of "Sparsity" from a set of
nonadaptive linear measurements. It has great potential of
application on various engineering fields such as coding and
information theory, signal processing, machine learning and others
(See papers in website \cite{zhang1} and reference herein). Sparse
signals contains much less information than their ambient dimension
suggests. Most of entries in its vector representation are zero (or
negligible). So it is possible to reconstruct original signal
approximately or even accurately using only a small number of linear
measurements. The measurements are of the form $\Phi{x}$ where
$\Phi$ is some $m{\times}N$ measurement matrix, $m\ll{N}$ and $x$ is
original signal. Clearly the process of sparse signal recovery could
be formulated as solving undetermined linear algebraic equation
$y=\Phi{x}$, where $y$ is measurement data. It is well-known that
undetermined equation has infinite solutions in general. But if we
focus on "Sparse" solution only, situation will be different.
Although the operation for finding the most sparse solution of
undetermined linear equation is NP-Hard commonly\cite{zhang2},
theoretical work in compressed sensing has shown that for certain
kinds of measurement matrices, it is possible when the number of
measurements m is nearly linear in the sparsity of original signal
\cite{zhang3}\cite{zhang4}.

The two major algorithmic approaches to sparse signal recovery are
based on $L_1$-minimization and on greedy methods (Matching
Pursuit). Finding the most sparse solution of undetermined linear
equation is a $L_0$ optimization problem:
\begin{equation}
\min\|x\|_0,{\quad}s.t{\quad}\|{\Phi}x-y\|_2\leq\delta,\label{label1}
\end{equation}
It could be solved by $L_1$ relaxation for some measurement matrices
$\Phi$\cite{zhang5}. That is, solving (\ref{label1}) is equivalent
to solving the following $L_1$ optimization problem
\begin{equation}
\min\|x\|_1,{\quad}s.t{\quad}\|{\Phi}x-y\|_2\leq\delta,\label{label2}
\end{equation}
where $\|x\|_1=|x_1|+\cdots+|x_n|$ for $x\in{C}^n$. Recently, more
stronger sufficient condition called Restricted Isometry Property
(RIP) on measurement matrix $\Phi$ to guarantee the equivalence of
(\ref{label1}) and (\ref{label2}) was also proposed \cite{zhang6}.
It was widely accepted that $L_1$-minimization (\ref{label2}) was
normal path to complete sparse signal recovery. (\ref{label2}) is
essentially a linear programming problem and technique of convex
optimization could be utilized to solve it effectively
\cite{zhang7}. The $L_1$-minimization method provides uniform
guarantees for sparse recovery. Once the measurement matrix $\Phi$
satisfies the restricted isometry condition, this method works
correctly for all sparse signals $x$. The method is also stable, so
it works for non-sparse signals such as those which are
compressible, as well as noisy signals. However, the method is based
on linear programming, and there is no strongly polynomial time
algorithm in linear programming \cite{zhang8}. But its efficiency
was questionable and most popular software package of convex
programming, such as cvx\cite{zhang9}, is hard to be used in
practical application for its low rate of convergence, especially
when dimension of target signal is large.

On the other hand, Greedy algorithms are quite fast, both
theoretically and experimentally. It runs by iterating in general.
Typically, On each iteration, the inner products of residue vector
$r$ with the columns of measurement matrix $\Phi$ is computed  and a
least squares problem is solved to obtain the estimation of original
signal on this iteration. It is hoped that the convergence of
iterating could be ensured and the estimator could tend to original
signal in fewer steps \cite{zhang10}.

The typical case of greedy algorithms for sparse recovery includes
Orthogonal Matching Pursuit (OMP) \cite{zhang11}, Regularized
Orthogonal Matching Pursuit (ROMP) \cite{zhang8} and compressive
sampling matching pursuit (CoSaMP) \cite{zhang12}. It was shown that
OMP recovered the sparse signal with high probability and had great
speed, but it would fail for some sparse signals and matrices
\cite{zhang13}. The development of ROMP provides a greedy algorithm
with uniform guarantees for sparse recovery, just as that provided
by $L_1$-minimization method. Furthermore, CoSaMP improves upon
these results and provides rigorous runtime guarantees. However,
there are one disadvantage for these two algorithms. Firstly,
sparsity level must be presented as prior parameter for algorithms.
But it is unknown in most practical scenario and must be guessed in
advance. Once the estimated sparsity level has large difference with
actual one, the error of algorithm will increase evidently (Although
this error could be analyzed theoretically \cite{zhang12}). The
problem becomes more severe in noisy environment. ROMP and CoSaMP
can't adapt their running process to noise condition when noise is
presented. Actually, noise is inevitable in engineering problem and
target signal of our recovery algorithms is always buried in it.
There exist few "Pure sparse" signal in real world. Because the
dimension of target signal is unknown, it is always given with some
margin to avoid possible missing. Hence it is hard to extract target
signal without including certain amount of noise. This not only has
influence on accuracy of algorithm, but also reduce the speed of
convergence for algorithms. In fact, some calculation is carried
through to estimate noise, however, which is useless at all.

A new greedy algorithm for sparse recovery is presented in this
paper. Compared with ROMP and CoSaMP, our new algorithm need no any
prior information on sparsity level of target signal. Furthermore,
it is a kind of "adaptive" algorithm which can inspect the existence
of noise and adjust the halting condition automatically based on
detailed state of noise. Besides that, it has uniformly guarantee
and good efficiency, just as ROMP or CoSaMP. Hence our new algorithm
is a better choice when signal with unknown sparsity level is to be
extracted (such as compressible signal) under noisy background. This
paper is organized as follows: In Section 2 we introduce our new
algorithm. Section 3 describes some consequences of the restricted
isometry property that pervade our analysis. The convergence of
theorem is also established for sparse signals in Sections 3.
Practical consideration for algorithm implementation is provided in
Section 4. Empirical performance and some numerical experiment is
described in Section 5. Finally, Section 7 presents overall
conclusion.

\section{Description of new algorithm}\label{sec1}

\subsection{Motivation}

The most difficult and important part of signal reconstruction is to
identify the locations of the components in the target signal. The
common approach adopted by most greedy algorithms is "Local
Approximating", that is, computing the inner products between
measurement vector $y$ and columns of measurement matrix $\Phi$. We
will obtain observation vector $u=\Phi^Hy$ and use $u$ as "Local
Approximation" (or "Proxy") of target signal $x$. Note that $\Phi$
is a dictionary and $v$ is sparse, so $y$ has a sparse
representation with respect to the dictionary. It is reasonable that
only a few entries of $u$ are remarkable, which imply the locations
of the components of $x$, and most of its entries are comparatively
small. Of course, the precondition for argument above is that $\Phi$
must satisfy some condition such as RIP. Intuitively, given sparsity
level $n$ of $x$, every $n$ columns form approximately an
orthonormal system. Therefore, every $n$ coordinates of the
observation vector u look like correlations of the measurement
vector $y$ with the orthonormal basis and therefore are close in the
Euclidean norm to the corresponding coefficients of $x$.

Popular greedy algorithms, including OMP, ROMP and CoSaMP, pay much
attention to observation vector $u$ and build their estimator of
location of components in $x$ based on $u$. OMP uses one biggest
coordinate of $u$. It is argued that using only the biggest couldn't
provides uniformly guarantee. So ROMP makes use of the $n$ biggest
coordinates of $u$, rather than just biggest one, and take a further
step of regularization to improve the performance of algorithm. It
should be noted that sparsity level $n$ is always unknown. CoSaMP
employs more coordinates of $u$, the $2n$ biggest, to avoid the
possible leakage of component in $x$. But $n$ must be guessed to be
input in ROMP or CoSaMP as important parameter. If guessed $n$ is
smaller than its true value, correct result can't be found; On the
other hand, if guessed $n$ is set to a very large value (maybe much
larger than true value) to ensure that all of entries of $x$ will
enter the view of algorithms, certain amount of noise will presented
in our calculation inevitably. How can we make a good guess for
sparsity level $n$ without any knowledge on its true value?

\subsection{AMOP Algorithm}

Our new approach, named Adaptive Orthogonal Matching Pursuit (AMOP),
chooses appropriate number of biggest entries of observation vector
$u$ by studying the fine feature of $u$. At each step, $u$ is
computed by $u=\Phi^Hr$ where $r$ is the residue vector of last
step. Unlike other algorithms, AMOP determines the estimation for
$n$ by analyzing the trend of entries in $u$ arranged by descend
order of their amplitude. That is, relative amplitude difference of
adjacent elements in above queue is calculated and a threshold is
set. Once the relative amplitude difference between $k$th and
$(k+1)$th element in ordered queue of entries in $u$ exceed
threshold, $k$ will be chosen as estimation of $n$.

Detailed description of AMOP is proposed as follows:

\mbox{}
\begin{algorithm}[Adaptive Orthogonal Matching Pursuit]\label{label8}
\mbox{}
\begin{itemize}
\item[] \textbf{Input:} Measurement matrix $\Phi\in{R}^{M{\times}N}$,
Measurement vector $y\in{R}^N$, Threshold $T$, $\epsilon$ and $K$.

\item[1] Let $r=y$, $\Omega=\emptyset$.

\item[2] Calculate $u=\Phi^Hr$ and $|u|=(|u_1|,\cdots,|u_N|)$.

\item[3] Arrange $|u|$ by descend order to obtain
\begin{equation}
|u|_d=(|u|_{[1]},|u|_{[2]},\cdots,|u|_{[N]}),\label{label3}
\end{equation}

\item[4] Determine index $k$ as follows, let $\beta=1$,
\[
k=\min\left\{i\in\{1,\cdots,N\}\left|\frac{|u|_{[k+1]}-|u|_{[k]}}{|u|_{[k]}}>T*\beta\right.\right\},
\]

\item[5] if $k>K$ and $\beta<0.1$, set $k=K$; else $\beta=\beta*0.9$, goto step (4);

\item[6] Update the set of indices by $\Omega=\Omega\cup\{[1],\cdots,[k]\}$.

\item[7] Solve least square problem
\[
\min_{\hat{x}}\|\Phi|_{\Omega}\hat{x}-y\|_2,
\]
where $\Phi|_{\Omega}$ is a submatrix of $\Phi$ composed of its
columns with index in $\Omega$.

\item[8] Calculate the residue vector of $r=y-\Phi\hat{x}$.

\item[9] If $|r|/|y|<\epsilon$, output $\hat{x}$ and $\Omega$, stop; else go
to step (2).
\end{itemize}

\end{algorithm}
\mbox{}

As input, the AOMP algorithm requires two adjustable parameter $T$
and $\epsilon$ besides matrix $\Phi$ and measurement vector $y$. But
it doesn't need sparsity level of target signal $x$ anymore, unlike
ROMP and CoSaMP. It can be extracted incidentally along with running
of algorithm. Furthermore, the number of components selected in the
step (4) is determined by algorithm itself automatically. It is easy
to understand that this number is critical for performance of
algorithm. Any manual setting will introduce extra error when
mismatch between prescribed value and actual situation of data
exists. So it is very necessary to let greedy algorithm of sparse
recovery be adaptive, just as AOMP.

Step (5) should be noted that it give AMOP algorithm more
flexibility and stability. If threshold $T$ is set too large so that
too much coordinates was selected in one iteration, algorithm is
prone to degrade or crash. For this we build a upper bound in step
(5) to prevent the crazy growing of number of chosen components. If
this bound is exceed, threshold $T$ will be adjusted to smaller
value to increase the possibility of components in $|u|$ in step (3)
to satisfy the condition in step (4). The importance of step (5) is
also illustrated in following section on analysis of convergence.

\section{Theoretical Analysis of Algorithmic Performance}

There are two kinds of iterative invariant of greedy algorithm for
sparse recovery deduced in the convergence analysis for ROMP and
CoSaMP respectively. As to CoSaMP, assume the sparsity level $s$ is
preliminary and
\[
v=\|x-x_s\|_2+\frac{1}{\sqrt{s}}\|x-x_s\|_1+\|e\|_2,
\]
where $x_s$ is s-sparse approximation for $x$ and $e$ is additive
noise, the following assertion could be proved \cite{zhang12}.
\begin{equation}
\|x-\alpha^{k+1}\|_2\leq\|x-\alpha^k\|_2+10v,
\end{equation}
where $\alpha^k$ is result of pruning step in CoSaMP. Hence it is
forced to be s-sparse. So this kind of iterative invariant is not
suitable for analysis of AMOP because the sparsity level of
intermediate result at each step in AMOP isn't fixed. However, the
iterative invariant in ROMP simply concerns with the percentage of
the coordinates in the newly selected set belong to the support of
target signal $x$. It is argued that the ratio above isn't lower
than 50\% with the help of regularization step \cite{zhang8}. We
will follow the idea in \cite{zhang8} to derive our result on
convergence of AMOP.

\subsection{Localization of Energy}

By induction on the iteration of AOMP, we study the gain in each
iteration. Losing no generality, suppose sparsity level of target
signal $x$ be $S$,  and $k$ coordinates is selected eventually in
this iteration. Then its percentage of energy of first $k$
components in queue (\ref{label3}) is
\begin{equation}
P=\frac{\sum_{n=1}^k|y|_{[n]}^2}{\sum_{n=1}^k|y|_{[n]}^2+\sum_{n=k+1}^K|y|_{[n]}^2}
\end{equation}
For the descend order of queue (\ref{label3}), we have
\begin{eqnarray}
P&\geq&\frac{k|y|_{[k]}^2}{\sum_{n=1}^k|y|_{[n]}^2+(K-k)|y|_{[k]}^2}\nonumber\\
&\geq&\frac{k|y|_{[k]}^2}{\sum_{n=0}^{k-1}(1-T)^{-2n}|y|_{[k]}^2+(K-k)(1-T)^2|y|_{[k]}^2}\nonumber\\
&=&\frac{k}{\sum_{n=0}^{k-1}(1-T)^{-2n}+(K-k)(1-T)^2}\nonumber\\
&=&\frac{k}{\frac{1-(1-T)^{-2k}}{1-(1-T)^{-2}}+(K-k)(1-T)^2}\label{label4}
\end{eqnarray}
It is easily seen that \ref{label4} achieves its minimum at $k=1$,
that is
\begin{equation}\label{label5}
P_{min}=\frac{1}{1+(K-1)(1-T)^2}
\end{equation}
Here $k=1$ means only the largest coordinates was chosen, which is
just the choice of OMP. So OMP is a special case of AMOP. According
to \cite{zhang8} Lemma 3.6, a large portion of energy of
unidentified part of target signal would be locked by queue
(\ref{label3}) and certain amount of energy would be reserved by
"regularization step" in ROMP by \cite{zhang8} lemma 3.8. In AMOP,
the "regularization step" is replaced by choosing $k$ largest
coordinates, So more energy would be identified in AMOP than in OMP
because more than one coordinates would be chosen in AMOP generally.
The ability of locking uncovered energy of target signal for AMOP
and ROMP is compared in Fig.\ref{fig1}

\begin{figure}[h]
  \centering
  \centerline{\epsfig{figure=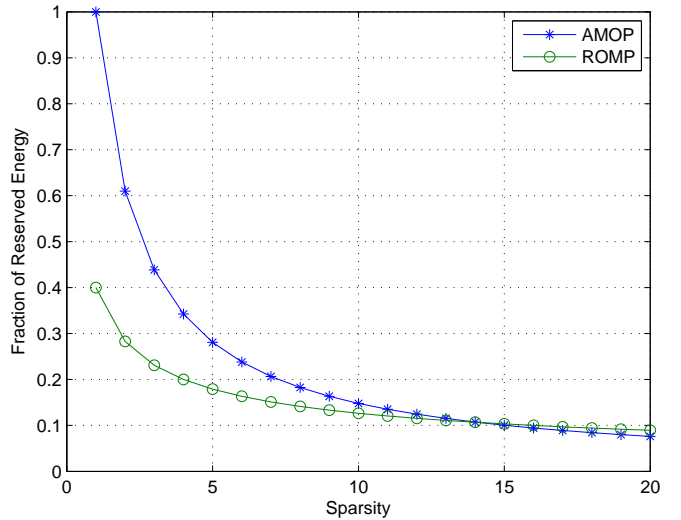,width=10cm}}
  \caption{Ability of Locking Uncovered Energy} \label{fig1}
\end{figure}

It is shown that ability of AMOP is superior to that of ROMP when
sparsity of target signal is small. Although the advantage becomes
vague when sparsity grows, AMOP still has relatively good
performance considering more coordinates would be chosen and
percentage of actual identified energy would be larger than
(\ref{label5}).

\subsection{Getting the "Correct" Support}

In \cite{zhang14} and \cite{zhang15}, the correctness of support of
solution for OGA algorithm under different noise scenario were
analyzed. Mutual coherence of matrix $\Phi$, overcomplete dictionary
system, was the key parameter for performance of greedy algorithm in
discussion therein. Here we will give analogous results for AMOP
under noisy condition based on RIC (Restricted Isometry
Constant)\cite{zhang3} of $\Phi$.

In practice the noise is unavoidable and it is always assumed that
ideal noiseless signal $y^0$ has sparse representation
$y^0={\Phi}x^0$, where the support of $x^0$ is very small. What we
can observe is noisy version $y=y^0+n$ where $\|n\|_2\leq\epsilon$.
Suppose $x^0$ be solution of
\begin{equation}
\min\|x\|_0,{\quad}s.t{\quad}{\Phi}x^0=y^0,\label{label6}
\end{equation}
$x^A$ be final result of AMOP, and
\begin{equation}
\mathcal{S}^0=\textrm{supp}(x^0),\quad\quad\mathcal{S}=\textrm{supp}(x^A)
\end{equation}
We argue that $\mathcal{S}{\subseteq}\mathcal{S}^0$ under certain
conditions on value distribution of target signal. That is, the
correctness of support of solution for AMOP can be guaranteed even
in noisy environment.

\begin{proposition}
if target signal $\{x_{\mathcal{S}}^0\}$ satisfy
\begin{equation}\label{label11}
\min_{j\in\mathcal{S}^0}\|x_j^0\|_2\geq2(\epsilon+\frac{\delta_K}{1-\delta_K}\sqrt{\frac{K}{2}}\max_{j\in\mathcal{S}^0}\|x_j^0\|_2),
\end{equation}
here matrix $\Phi$ is supposed to has $K$-order restricted isometry
constant $\delta_K$, $K$ is twice of sparsity level of target signal
$x^0$. Then $\mathcal{S}{\subseteq}\mathcal{S}^0$ holds throughout
the iteration process of AMOP unless all the coordinates in
$\mathcal{S}^0$ were chosen.
\end{proposition}

\begin{proof}

We proceed by induction. $\mathcal{S}{\subseteq}\mathcal{S}^0$ holds
at beginning of step 1 in AMOP initially for
$\mathcal{S}=\emptyset$. Assume it is true at beginning of step 1 in
given iteration, we prove it is still true at beginning of step 1 in
next iteration. Consider case of $\mathcal{S}{\subset}\mathcal{S}^0$
, we have
\begin{equation}
r=\Phi_{\mathcal{S}}x_{\mathcal{S}}-y,
\end{equation}
It is trivial that
\begin{equation}\label{label7}
\Phi_{\mathcal{S}}x_{\mathcal{S}}^0+
\Phi_{\mathcal{S}^0\setminus\mathcal{S}}x_{\mathcal{S}^0\setminus\mathcal{S}}^0-y^0
=\Phi_{\mathcal{S}^0}x^0_{\mathcal{S}^0}-y^0=0,
\end{equation}
because $x_{\mathcal{S}}$ is the solution of least square
optimization on step 7 in AMOP,
\begin{equation}
x_{\mathcal{S}}=(\Phi_{\mathcal{S}}^T\Phi_{\mathcal{S}})^{-1}\Phi_{\mathcal{S}}^Ty,
\end{equation}
Multiply
$(\Phi_{\mathcal{S}}^T\Phi_{\mathcal{S}})^{-1}\Phi_{\mathcal{S}}^T$
on two side of \ref{label7},
\begin{equation}
x_{\mathcal{S}}^0+
(\Phi_{\mathcal{S}}^T\Phi_{\mathcal{S}})^{-1}\Phi_{\mathcal{S}}^T
\Phi_{\mathcal{S}^0\setminus\mathcal{S}}x_{\mathcal{S}^0\setminus\mathcal{S}}^0
-(\Phi_{\mathcal{S}}^T\Phi_{\mathcal{S}})^{-1}\Phi_{\mathcal{S}}^Ty_0=0,
\end{equation}
we have
\begin{eqnarray}
r&=&\Phi_{\mathcal{S}}x_{\mathcal{S}}-y-\Phi_{\mathcal{S}}x_{\mathcal{S}}^0+
\Phi_{\mathcal{S}^0\setminus\mathcal{S}}x_{\mathcal{S}^0\setminus\mathcal{S}}^0-y^0\nonumber\\
&=&(\Phi_{\mathcal{S}}(\Phi_{\mathcal{S}}^T\Phi_{\mathcal{S}})^{-1}\Phi_{\mathcal{S}}^T-I)(y-y^0)\nonumber\\
&+&\Phi_{\mathcal{S}}(\Phi_{\mathcal{S}}^T\Phi_{\mathcal{S}})^{-1}\Phi_{\mathcal{S}}^T\Phi_{\mathcal{S}^0\setminus\mathcal{S}}x_{\mathcal{S}^0\setminus\mathcal{S}}^0\nonumber\\
&-&\Phi_{\mathcal{S}^0\setminus\mathcal{S}}x_{\mathcal{S}^0\setminus\mathcal{S}}^0
\end{eqnarray}
where $I$ is the identity matrix.

For $j\notin\mathcal{S}$, the norm of $\phi_j^Tr$ is estimated and
bounded. Firstly, because
$\Phi_{\mathcal{S}}(\Phi_{\mathcal{S}}^T\Phi_{\mathcal{S}})^{-1}\Phi_{\mathcal{S}}^T-I$
is the projection matrix of orthogonal complement of subspace
spanned by $\Phi_{\mathcal{S}}$, its 2-norm is 1. So
\begin{eqnarray}
&&\|\phi_j^T(\Phi_{\mathcal{S}}(\Phi_{\mathcal{S}}^T\Phi_{\mathcal{S}})^{-1}\Phi_{\mathcal{S}}^T-I)(y-y^0)\|_2\nonumber\\
&\leq&\|\phi_j^T\|_2\|(\Phi_{\mathcal{S}}(\Phi_{\mathcal{S}}^T\Phi_{\mathcal{S}})^{-1}\Phi_{\mathcal{S}}^T-I)\|_2\|(y-y^0)\|_2\nonumber\\
&=&1\ast1\ast\epsilon=\epsilon,
\end{eqnarray}
Secondly,
\begin{eqnarray}
&&\|\phi_j^T\Phi_{\mathcal{S}}(\Phi_{\mathcal{S}}^T\Phi_{\mathcal{S}})^{-1}\Phi_{\mathcal{S}}^T\Phi_{\mathcal{S}^0\setminus\mathcal{S}}x_{\mathcal{S}^0\setminus\mathcal{S}}^0\|_2\nonumber\\
&=&\|\phi_j^T\Phi_{\mathcal{S}}\|_2\|(\Phi_{\mathcal{S}}^T\Phi_{\mathcal{S}})^{-1}\|_2
\|\Phi_{\mathcal{S}}^T\Phi_{\mathcal{S}^0\setminus\mathcal{S}}\|_2\|x_{\mathcal{S}^0\setminus\mathcal{S}}^0\|_2,\label{label9}
\end{eqnarray}
Because $\sharp\{\mathcal{S}{\cup}\mathcal{S}^0\}\leq{2S}$ and
$K=2S$, According to \cite{zhang12}, proposition 3.2, we have
\begin{equation}
\|\Phi_{\mathcal{S}}^T\Phi_{\mathcal{S}^0\setminus\mathcal{S}}\|_2\leq\delta_K
\end{equation}
and
\begin{equation}
\|\phi_j^T\Phi_{\mathcal{S}}\|_2\leq\delta_K
\end{equation}
On the other hand, according to definition of RIC, we obtain
\begin{equation}
\frac{1}{1+\delta_K}\leq\|(\Phi_{\mathcal{S}}^T\Phi_{\mathcal{S}})^{-1}\|_2\leq\frac{1}{1-\delta_K}
\end{equation}
Hence
\begin{eqnarray}
&&\|\phi_j^T\Phi_{\mathcal{S}}(\Phi_{\mathcal{S}}^T\Phi_{\mathcal{S}})^{-1}\Phi_{\mathcal{S}}^T\Phi_{\mathcal{S}^0\setminus\mathcal{S}}x_{\mathcal{S}^0\setminus\mathcal{S}}^0\|_2\nonumber\\
&\leq&\frac{\delta_K^2}{1-\delta_K}\|x_{\mathcal{S}^0\setminus\mathcal{S}}^0\|_2,\label{label10}
\end{eqnarray}
Thirdly, by analogous deduction, if $j\notin\mathcal{S}^0$,
\begin{equation}
\|\phi_j^T\Phi_{\mathcal{S}^0\setminus\mathcal{S}}x_{\mathcal{S}^0\setminus\mathcal{S}}^0\|_2\leq\delta_K\|x_{\mathcal{S}^0\setminus\mathcal{S}}^0\|_2
\end{equation}
Summarize the results above, we have for $j\notin\mathcal{S}^0$,
\begin{eqnarray}
\|\phi_j^Tr\|_2&\leq&\epsilon+(\frac{\delta_K^2}{1-\delta_K}+\delta_K)\|x_{\mathcal{S}^0\setminus\mathcal{S}}^0\|_2,\nonumber\\
&\leq&\epsilon+\frac{\delta_K}{1-\delta_K}\|x_{\mathcal{S}^0\setminus\mathcal{S}}^0\|_2,\nonumber\\
&\leq&\epsilon+\frac{\delta_K}{1-\delta_K}\sqrt{\frac{K}{2}}\max_{j\in\mathcal{S}^0\setminus\mathcal{S}}\|x_j^0\|_2\nonumber\\
&\leq&\epsilon+\frac{\delta_K}{1-\delta_K}\sqrt{\frac{K}{2}}\max_{j\in\mathcal{S}^0}\|x_j^0\|_2,
\end{eqnarray}

Lower bound for $\|\phi_j^Tr\|_2$ is considered similarly. For
$j\in\mathcal{S}^0\setminus\mathcal{S}$,
\begin{eqnarray}
&&\|\phi_j^T\Phi_{\mathcal{S}^0\setminus\mathcal{S}}x_{\mathcal{S}^0\setminus\mathcal{S}}^0\|_2\nonumber\\
&=&\|x_{j}^0+\phi_j^T\Phi_{\mathcal{S}^0\setminus(\mathcal{S}\cup\{j\})}x_{\mathcal{S}^0\setminus(\mathcal{S}\cup\{j\})}^0\|_2,\nonumber\\
&\geq&\|x_{j}^0\|_2-\|\phi_j^T\Phi_{\mathcal{S}^0\setminus(\mathcal{S}\cup\{j\})}x_{\mathcal{S}^0\setminus(\mathcal{S}\cup\{j\})}^0\|_2\nonumber\\
&\geq&\|x_{j}^0\|_2-\delta_K\|x_{\mathcal{S}^0\setminus(\mathcal{S}\cup\{j\})}^0\|_2\nonumber\\
&\geq&\|x_{j}^0\|_2-\delta_K\|x_{\mathcal{S}^0\setminus\mathcal{S}}^0\|_2
\end{eqnarray}
Hence
\begin{eqnarray}
\|\phi_j^Tr\|_2&\geq&\|x_{j}^0\|_2-\delta_K\|x_{\mathcal{S}^0\setminus\mathcal{S}}^0\|_2
-\epsilon-\frac{\delta_K^2}{1-\delta_K}\|x_{\mathcal{S}^0\setminus\mathcal{S}}^0\|_2\nonumber\\
&=&\|x_{j}^0\|_2-\epsilon-\frac{\delta_K}{1-\delta_K}\|x_{\mathcal{S}^0\setminus\mathcal{S}}^0\|_2\nonumber\\
&\geq&\|x_{j}^0\|_2-\epsilon-\frac{\delta_K}{1-\delta_K}\sqrt{\frac{K}{2}}\max_{j\in\mathcal{S}^0}\|x_j^0\|_2,
\end{eqnarray}
if target signal $\{S^0\}$ satisfy (\ref{label11}), we have
\begin{equation}
\min_{j\in\mathcal{S}^0\setminus\mathcal{S}}\|\phi_j^Tr\|_2\geq\max_{j\notin\mathcal{S}^0}\|\phi_j^Tr\|_2,
\end{equation}
That is to say, $\mathcal{S}\subseteq\mathcal{S}^0$ will hold
throughout the iteration process of AMOP unless all the coordinates
in $\mathcal{S}^0$ were chosen.
\end{proof}

It is argued that value distribution of target signal, power of
noise and RIC of matrix $\Phi$ are all critical to performance of
greedy algorithm. The proposition above gives a general condition
for correctness of support of solutions for a large class of greedy
algorithms (Not just AMOP) which use inner product between residue
$r$ and dictionary vectors (columns of $\Phi$) to obtain information
of support of target signal. It seems that condition (\ref{label11})
is too restricted. But it is easy to see from (\ref{label11}) that
"Dynamic Scope" of target signal (that is, the norm difference
between the elements with maximal and minimal norm) depends on RIC
$\delta_K$ of matrix $\Phi$ and noise power $\epsilon$. Consider the
requirement on RIC in ROMP, which is $0.03/\sqrt{\log(s)}$ with $s$
is sparsity level of target signal according to \cite{zhang8},
Theorem 3.1, we write (\ref{label11}) as
\begin{equation}
\min_{j\in\mathcal{S}^0}\|x_j^0\|_2=\frac{0.06\sqrt{s}}{\sqrt{\log(s)}-0.03}+2\epsilon,
\end{equation}
with maximum is normalized to 1. It is depicted in Fig.\ref{fig2}
for noise level is 0.1(SNR is 20dB). When sparsity level is small,
target signals with considerable 'Dynamic Scope' are guaranteed to
have good performance in greedy algorithms. The restriction on
'Dynamic Scope' of target signal becomes tighter gently when
sparsity level increases. The actual number of chosen coordinates in
iteration step of AMOP in practical scenario is smaller than
sparsity level in general. So AMOP could choose correct coordinates
in most cases. This assertion will be illustrated further in
numerical experiments.

\begin{figure}[h]
  \centering
  \centerline{\epsfig{figure=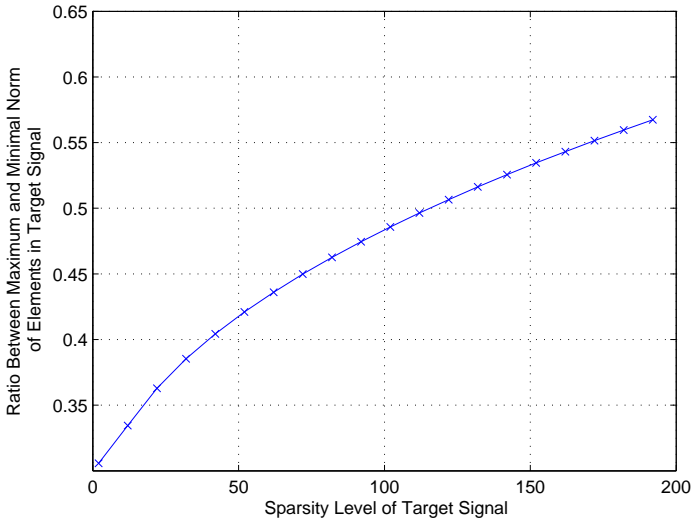,width=10cm}}
  \caption{Ratio Between Maximum and Minimum of Element Norm of Target Signal} \label{fig2}
\end{figure}

\section{Practical Consideration For Algorithm Implementation}

\subsection{Least Square via Orthogonalization}

For efficient implementation of AMOP, it is important to design a
appropriate computational scheme with low complexity and good
numerical stability to calculate the solution of least square
problem in step 7 in AMOP. It should be noticed that AMOP is
incremental, that is, in each iteration some new coordinates were
selected and none of previously chosen coordinates was excluded,
unlike CoSaMP. So it is possible to construct a recursive algorithm
to solve the least square problem.

Assume on the first iteration several coordinates were chosen and
denote the set of corresponding columns of matrix $\Phi$ as $A_1$.
Then observation vector $y$ would be linearly approximating with
vectors in $A_1$ and coefficients $\hat{x}_1$are computed by solving
the following equation
\begin{equation}
A_1^TA_1\hat{x}_1=A^Ty,
\end{equation}
ordinary solver with good numerical performance such as QR
decomposition or Singular value decomposition could be used to
compute $\hat{x}_1$. From geometrical point of view, calculation of
$\hat{x}_1$ is equivalent to project $y$ onto subspace spanned by
$A_1$. We wrote it intuitively as
\begin{equation}
\hat{x}_1=y|A_1,
\end{equation}
On the next iteration, a set $A_2$ of columns of matrix $\Phi$ with
respect to newly chosen coordinates would be added to least square
regression. The projection became
\begin{equation}
\hat{x}_2=y|\{A_1,A_2\},
\end{equation}
It is well-known that orthogonalization could simplify the
calculation for projection onto subspace spanned by two mutually
orthogonal subspace could be regarded as sum of projections on each
one. So $A_2$ was written as
\begin{equation}
A_2=A_2^1{\oplus}A_2^2,
\end{equation}
where $A_2^1$ was projection of $A_2$ onto $A_1$ and $A_2^2$ is
orthogonal to $A_1$. Hence
\begin{eqnarray}
\hat{x}_2&=&y|\{A_1,A_2\}=y|\{A_1,A_2^2\}\nonumber\\
&=&y|A_1+y|A_2^2=\hat{x}_1+y|A_2^2,\label{label12}
\end{eqnarray}
This could be accomplished with Gram-Schmidt orthogonalization
procedure. Without loss of generality, suppose
\begin{eqnarray}
A_1&=&(\phi_1,\phi_2,\cdots,\phi_k),\nonumber\\
A_2&=&(\phi_{k+1},\phi_{k+2},\cdots,\phi_n),\nonumber
\end{eqnarray}
then
\begin{equation}
A_1{\cup}A_2=\{\phi_1,\phi_2,\cdots,\phi_n\},
\end{equation}
using following procedure
\begin{eqnarray}
U_1&=&\phi_1,\nonumber\\
U_k&=&\phi_k-\sum_{m=1}^{k-1}\frac{\langle\phi_k,U_m\rangle}{\langle{U_m,U_m}\rangle}U_m,\label{label13}
\end{eqnarray}
where $\langle\bullet\rangle$ denotes the inner product of vectors
in Euclidean space. We can obtain
\begin{eqnarray}
B_1&=&(U_1,U_2,\cdots,U_k),\nonumber\\
B_2&=&(U_{k+1},U_{k+2},\cdots,U_n).\nonumber
\end{eqnarray}
The projection in (\ref{label12}) could be written as
\begin{equation}
y|A_2^2=y|B_2=y|U_{k+1}+\cdots+y|U_{n}
\end{equation}
The numerical stability of Gram-Schmidt orthogonalization procedure
could be improved further \cite{zhang16}. Instead of computing the
vector $U_k$ as (\ref{label13}), it is computed as
\begin{eqnarray}
U_k^{(1)}&=&\phi_k-\frac{\langle\phi_k,U_1\rangle}{\langle{U_1,U_1}\rangle}U_1,\nonumber\\
U_k^{(2)}&=&U_k^{(1)}-\frac{\langle{U_k^{(1)},U_2}\rangle}{\langle{U_2,U_2}\rangle}U_2,\nonumber\\
&\cdots&\nonumber\\
U_k^{(k-2)}&=&U_k^{(k-3)}-\frac{\langle{U_k^{(k-3)},U_k^{(k-2)}}\rangle}{\langle{U_k^{(k-2)},U_k^{(k-2)}}\rangle}U_k^{(k-2)},\nonumber\\
U_k&=&U_k^{(k-2)}-\frac{\langle{U_k^{(k-1)},U_k^{(k-2)}}\rangle}{\langle{U_k^{(k-1)},U_k^{(k-1)}}\rangle}U_k^{(k-1)},
\end{eqnarray}
This approach gives the same result as the original formula in exact
arithmetic, but it introduces smaller errors in finite-precision
arithmetic.

There are one points worth mentioning. Although other
orthogonalization algorithms such as Householder transformations or
Givens rotations are more stable than the stabilized Gram-Schmidt
process, they produce all the orthogonalized vectors only at the
end. On the contrary, the Gram-Schmidt procedure produces the $j$th
orthogonalized vector after the $j$th iteration and this makes it
the only choice for iterative algorithm like AMOP.

\subsection{Resource Requirements}

AMOP was designed to be a practical method for sparse signal
recovery. The main barrier for algorithm efficiency is least square
problem in step 7 of AMOP. Using recursive orthogonalization
procedure above could mitigate computational burden of AMOP
dramatically. Furthermore, the orthogonalization technique has the
additional advantage that they only interact with the matrix $\Phi$
through its action on vectors. In fact, it only concern with the
inner products and additions of columns of matrix $\Phi$. It follows
that the algorithm performs better when this sampling matrix has a
fast matrix-vector multiply, such as on parallel computational
platforms. On the other hand, less memory consumption is another
advantage of recursive orthogonalization based least square. In
fact, direct method such as SVD and QR have storage cost $O(km)$,
where $k$ is the number of chosen coordinates in each iteration and
$m$ is row number of matrix $\Phi$. It is too huge for large scale
problems. But for AMOP, only one vector need be put in memory in
recursive orthogonalization calculation and storage cost is $O(m)$.
It is more suitable for implemented with VLSI circuit.

We estimate the time complexity of main steps in AMOP as follows:
\begin{itemize}
\item[] \textbf{Step 2:} In this step, the inner products of residue
vector and columns of matrix $\Phi$ is computed as proxy for support
of target signal. Its cost is bounded by matrix-vector multiply
$\Phi^Tr$, which is $O(mN)$ with standard multiply operation or
$O(\mathcal{L})$ for fast matrix-vector multiply.

\item[] \textbf{Step 3:} According to standard textbook on
algorithms \cite{zhang17}, the expected time for selecting largest
$s$ entries in vector with dimension $N$ is $O(KN)$. Using efficient
schemes such as QuickSort or HeapSort, a fully sorting of vector
could be completed with expected time cost $O(N\log{N})$ and largest
$s$ entries could be selected directly, which is faster n some
situation.

\item[] \textbf{Step 4 \& 5:} Certain amount of support of target
signal would be identified in these two steps. Although sometimes
the threshold needs to be adjusted according to step 5 and several
cycles of operations may be necessary, the total cost is still
$O(K)$.

\item[] \textbf{Step 7:} The main advantage of AMOP is recursive
orthogonalization based implementation of least square problem.
Inner products of vectors are involved in orthogonalization and
occupy much of computational resource which can be implemented
efficiently by matrix-vector multiply. The cost is $O(mK)$ with
standard multiply operation or $O(\mathcal{L})$ for fast
matrix-vector multiply.
\end{itemize}

Table 1 summarizes the discussion above in standard multiply
operation and fast matrix-vector multiply with cost $L$,
respectively.

\begin{table}
\renewcommand{\arraystretch}{2}
\caption{Time Complexity of AMOP} \label{table1} \centering
\begin{tabular}{c||c|c}
\hline\hline
\bfseries Step & \bfseries Standard & \bfseries Fast\\
\hline\hline
2 & $O(mN)$ & $O(\mathcal{L})$ \\
\hline
3 & $O(N\log{N})$ & $O(N\log{N})$ \\
\hline
4\&5 & $O(K)$ & $O(K)$ \\
\hline
7 & $O(mK)$ & $O(K)$ \\
\hline\hline
Total & $O(mN)$ & $O(\mathcal{L})$ \\
\hline\hline
\end{tabular}
\end{table}

Storage cost for AMOP is also considered for showing its
practicability. Aside from the storage required by the sampling
matrix, AMOP algorithm constructs only one vector of length N as the
signal proxy. The sample vectors have length m, so they require
$O(m)$ storage. The signal approximations require at most $O(s)$
storage. Similarly, the index sets that appear require only $O(s)$
storage. The total storage is at worst $O(N)$.

\section{Numerical Experiment}

In this section some numerical experiments were conducted to
illustrate the performance of signal recovery of AOMP. There are
three factors to be considered in numerical testing of AMOP:
construction of $\Phi$ matrix, value distribution of target signal
and SNR condition which will be examined in our experiments.

\subsection{Construction of Matrix $\Phi$}

The property of measurement matrix $\Phi$ is critical to performance
of any greedy algorithm for sparse recovery. As indicated in section
on theoretical analysis, Its RIC has direct influence on probability
of recovery of algorithms. Here, several kinds of matrix $\Phi$ were
built and utilized to test the performance of AMOP, including
well-known Gaussian random matrix, Bernoulli random matrix and
random Fourier matrix.

The target signal was set to be flat and no noise was added in, then
500 independent trials were performed. Figure \ref{fig3}-\ref{fig5}
depicts the percentage (from the 500 trials) of sparse flat signals
that were reconstructed exactly when $m\times{N}$ Gaussian random
matrix was chosen as measurement matrix $\Phi$ . This plot was
generated with $N=256$ for various levels of sparsity $S$. The
horizontal axis represents the number of measurements $m$, and the
vertical axis represents the exact recovery percentage.

As comparison, recovery percentage of algorithm CoSaMP is also given
under the same setting. Standard CoSaMP needs sparsity of target
signal as its important prior knowledge and it is widely regarded as
one of main drawbacks of CoSaMP. In our experiment, sparsity of
target signal was given to CoSaMP as input parameter to guarantee
the power of CoSaMP to be exploited fully.

\begin{figure}[h]
  \centering
  \centerline{\epsfig{figure=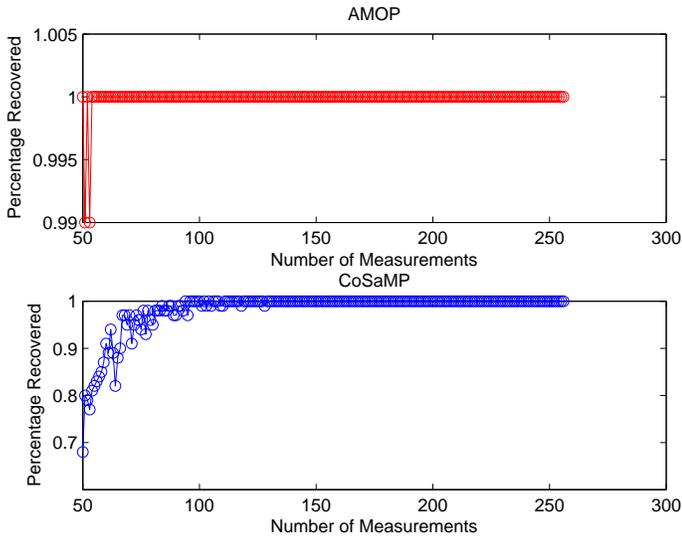,width=10cm}}
  \caption{Recovery Percentage of Signal with Sparsity 4 with Gaussian Matrix} \label{fig3}
\end{figure}

\begin{figure}[h]
  \centering
  \centerline{\epsfig{figure=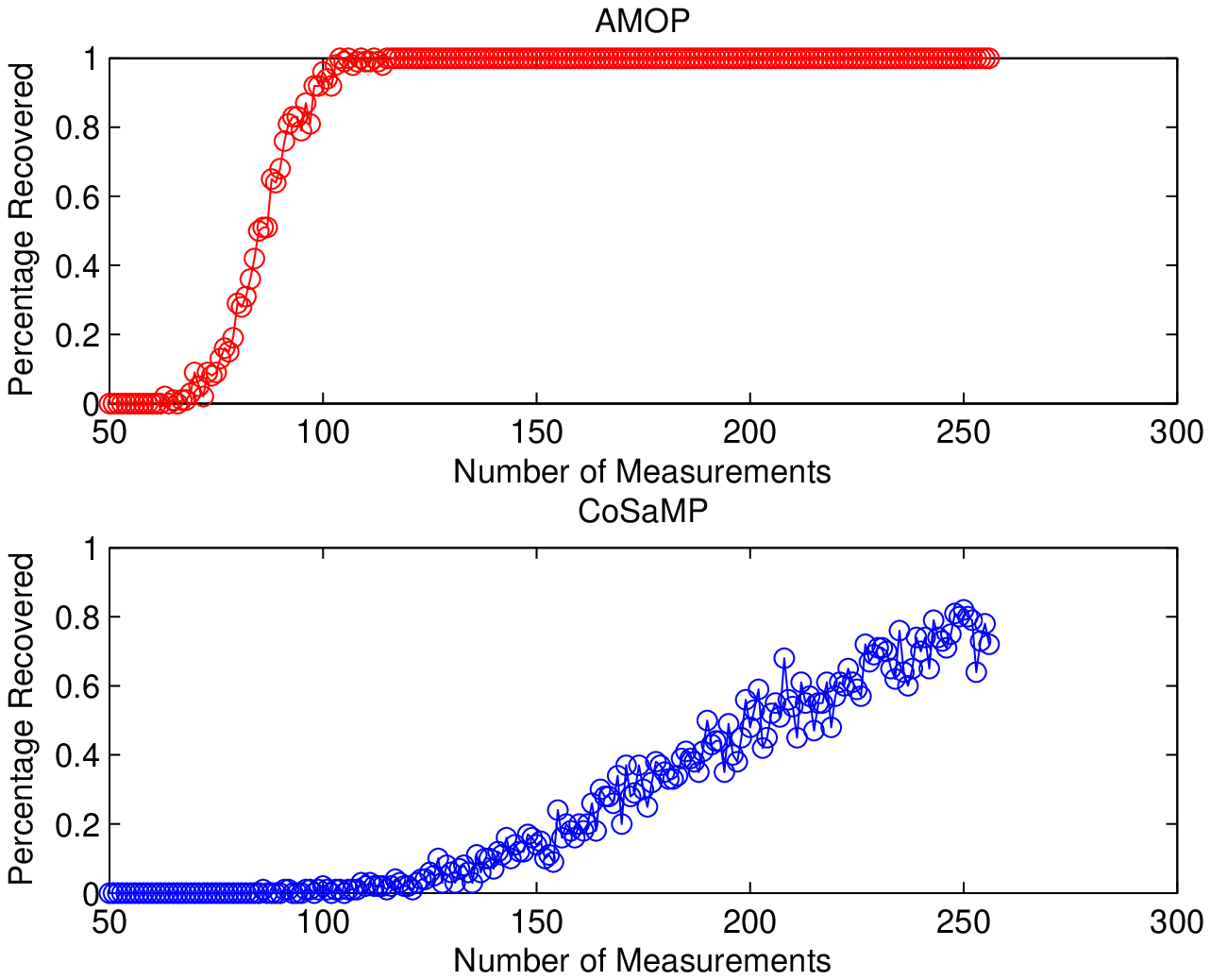,width=10cm}}
  \caption{Recovery Percentage of Signal with Sparsity 20 with Gaussian Matrix} \label{fig4}
\end{figure}

\begin{figure}[h]
  \centering
  \centerline{\epsfig{figure=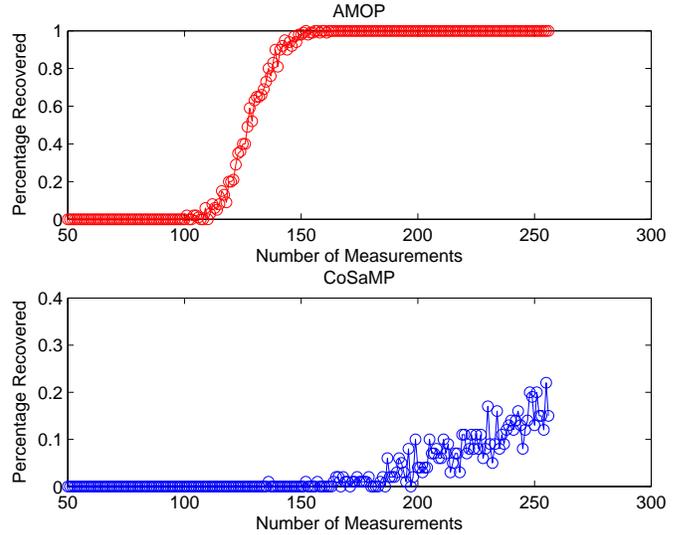,width=10cm}}
  \caption{Recovery Percentage of Signal with Sparsity 36 with Gaussian Matrix} \label{fig5}
\end{figure}

It should be noted that even the sparsity of target signal is known
beforehand (which is impossible in practice), recovery percentage of
CoSaMP is lower than that of AMOP. Especially when sparsity was
relatively large, performance of CoSaMP degenerated very rapidly. On
the contrary, the behavior of AMOP was very stable. According to
well-known theoretical result of Compressed Sensing, for Gaussian
random measurements matrix $\Phi$, if row number $m$, column number
$N$ and sparsity $S$ satisfies
\begin{equation}\label{label20}
m{\geq}CS\log(N),
\end{equation}
where $C$ is a constant independent of $S$, then the probability of
recovery failure is exponentially small \cite{zhang3}. Our
experiment result indicates that for AMOP, the value of constant $C$
is about $2$ for Gaussian measurement matrix.

\begin{figure}[h]
  \centering
  \centerline{\epsfig{figure=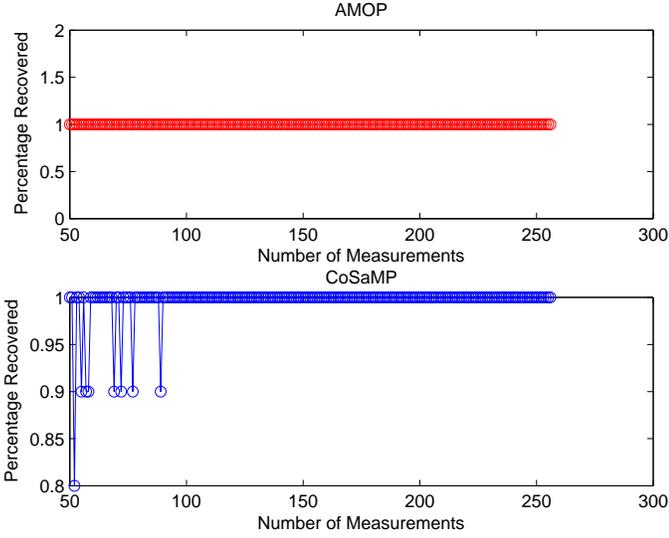,width=10cm}}
  \caption{Recovery Percentage of Signal with Sparsity 4 with Bernoulli Matrix} \label{fig6}
\end{figure}

\begin{figure}[h]
  \centering
  \centerline{\epsfig{figure=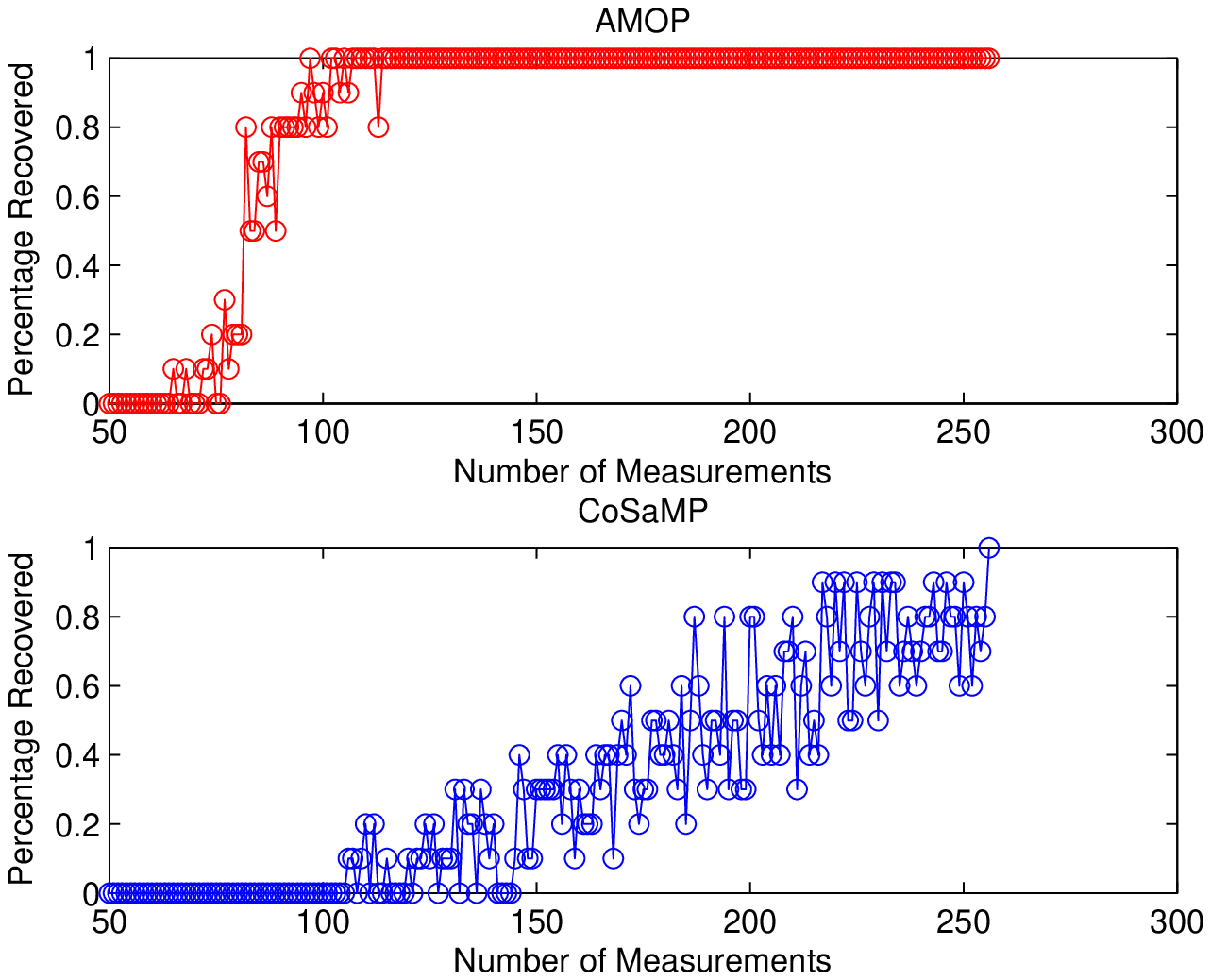,width=10cm}}
  \caption{Recovery Percentage of Signal with Sparsity 20 with Bernoulli Matrix} \label{fig7}
\end{figure}

\begin{figure}[h]
  \centering
  \centerline{\epsfig{figure=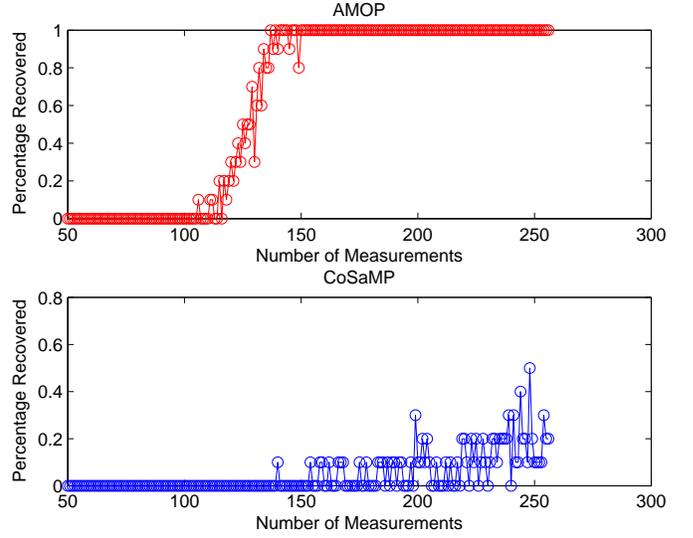,width=10cm}}
  \caption{Recovery Percentage of Signal with Sparsity 36 with Bernoulli Matrix} \label{fig8}
\end{figure}

Figure \ref{fig6}-\ref{fig8} depicts corresponding result for
Bernoulli random measurement matrix, which is analogous to Gaussian
case. It had been proved that condition (\ref{label20} is also
sufficient for overwhelming probability of successful recovery for
binary Bernoulli measurement matrix \cite{zhang18}. It is observed
that the constant $C$ for AMOP in Bernoulli case is probably the
same as that in Gaussian case.

\begin{figure}[h]
  \centering
  \centerline{\epsfig{figure=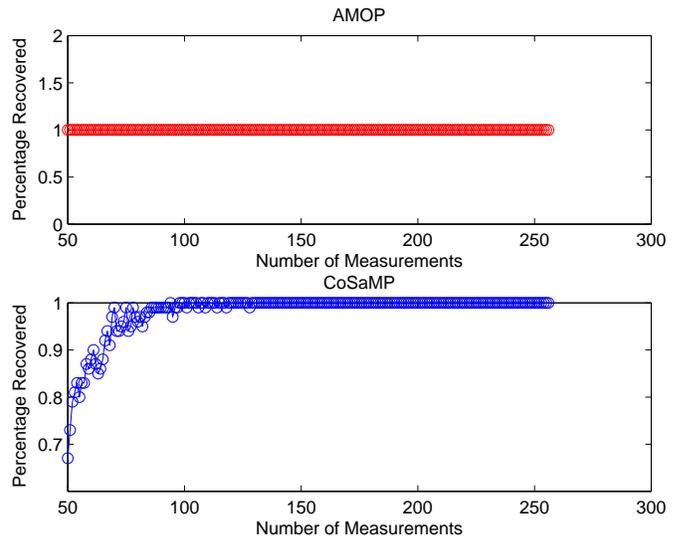,width=10cm}}
  \caption{Recovery Percentage of Signal with Sparsity 4 with Fourier Matrix} \label{fig9}
\end{figure}

\begin{figure}[h]
  \centering
  \centerline{\epsfig{figure=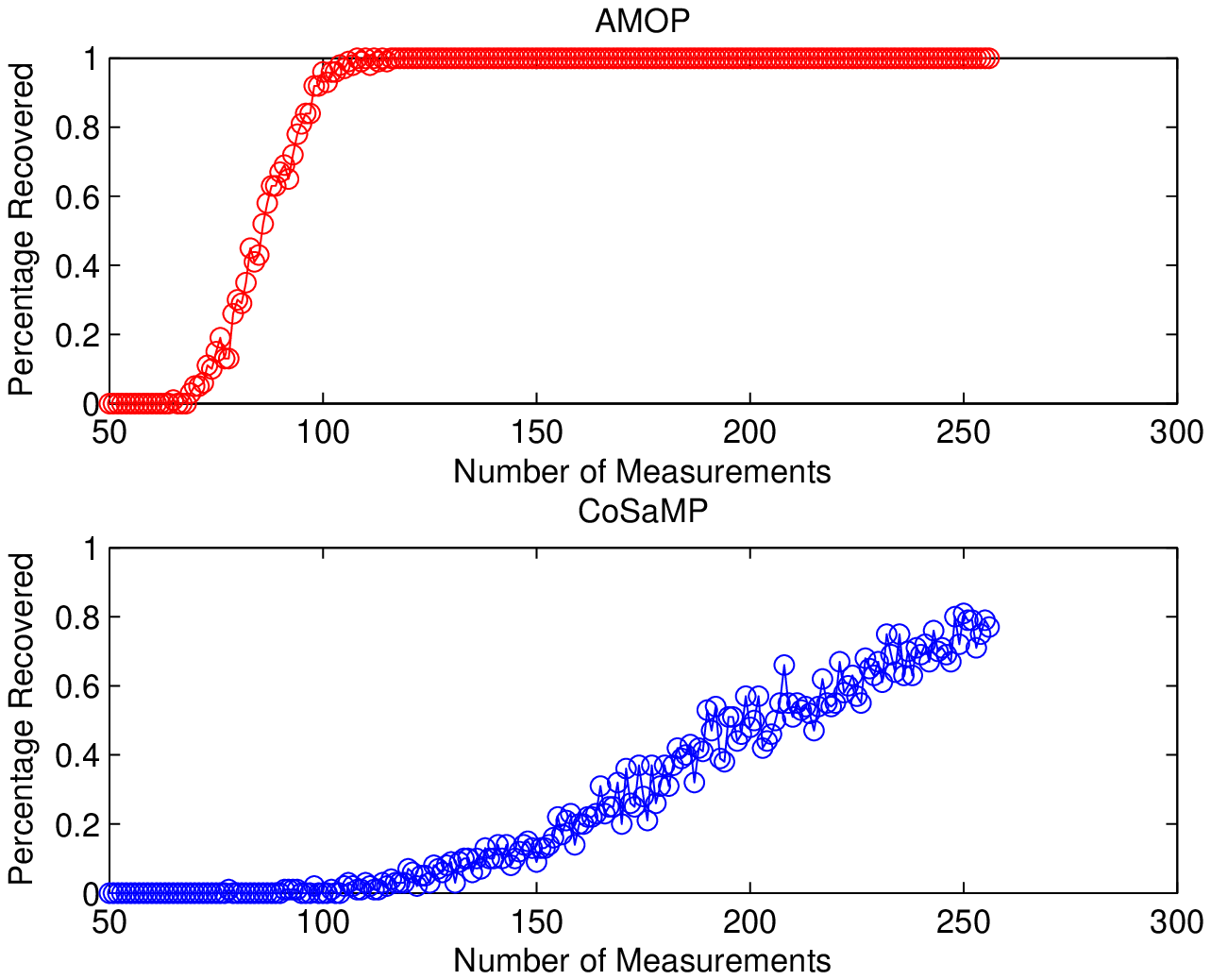,width=10cm}}
  \caption{Recovery Percentage of Signal with Sparsity 20 with Fourier Matrix} \label{fig10}
\end{figure}

\begin{figure}[h]
  \centering
  \centerline{\epsfig{figure=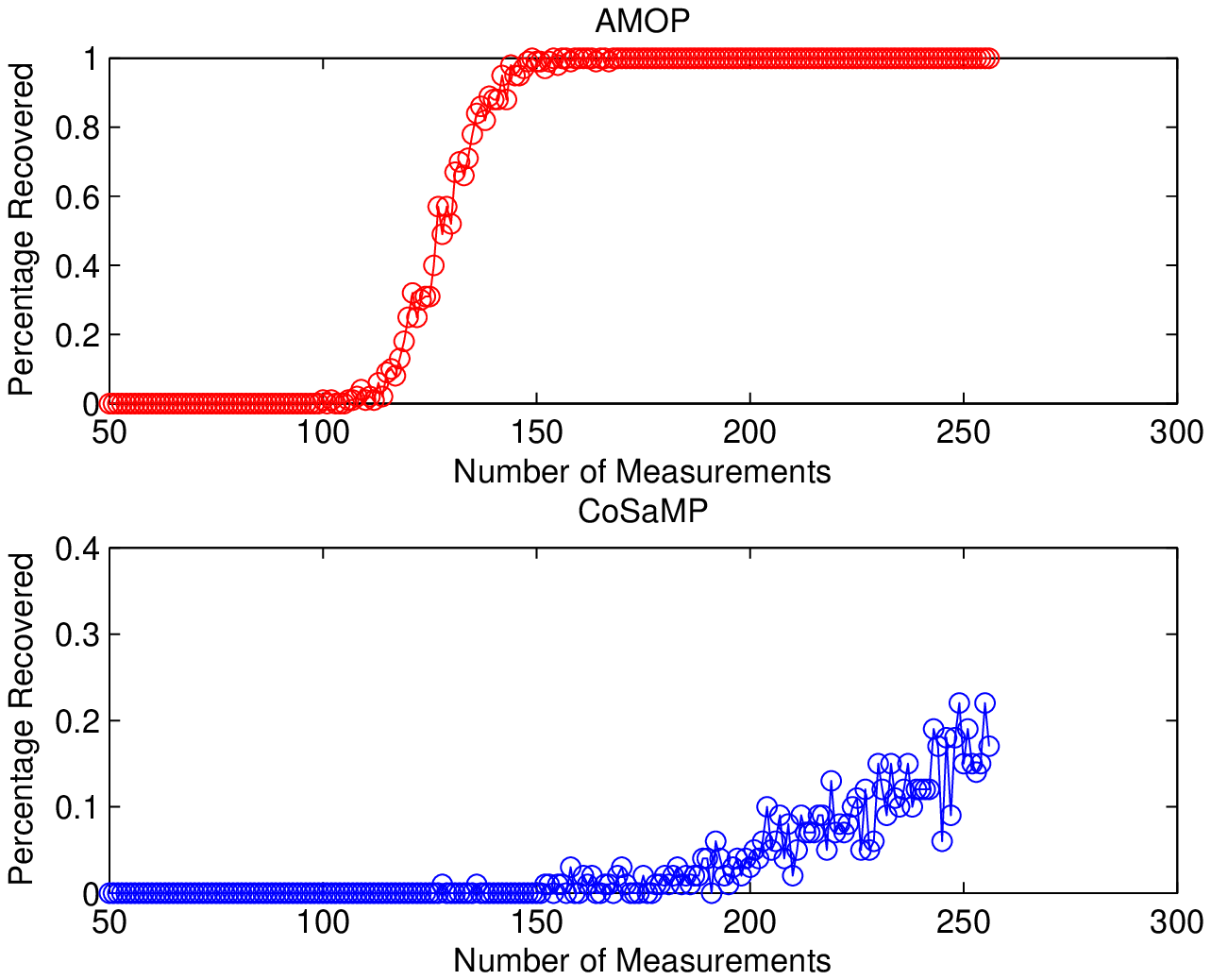,width=10cm}}
  \caption{Recovery Percentage of Signal with Sparsity 36 with Fourier Matrix} \label{fig11}
\end{figure}

Figure \ref{fig9}-\ref{fig11} depicts corresponding result for
Fourier random measurement matrix. Somewhat surprisingly, it is
similar to that of Gaussian and Bernoulli case. To our knowledge,
the best known bounds on size of measurements in Fourier case is
given by \cite{zhang19}
\begin{equation}
m{\geq}CS(\log(N))^4,
\end{equation}
which is conjectured to be the same as \ref{label20} \cite{zhang3}
but there exists no strict theoretical proof until now. Our
experiment result verified this conjecture in some extent
indirectly.

\subsection{Value Distribution of Target Signal}

Pure flat signal is rarely seen in practical engineering
application. So it is necessary to investigate the performance of
sparse recovery algorithms on non-flat target signal. There are two
cases to be studied. One is piecewise flat signal which is common in
various fields of imaging, such as optical, microwave and magnetic
resonance. The result is depicted in Figure \ref{fig12} to
\ref{fig14}. Here the measurement matrix is fixed to Gaussian random
matrix.

\begin{figure}[h]
  \centering
  \centerline{\epsfig{figure=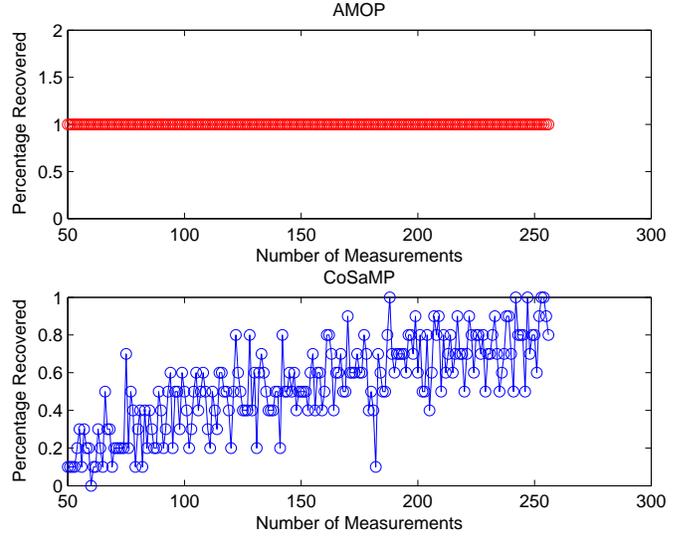,width=10cm}}
  \caption{Recovery Percentage of Piecewise Flat Signal with Sparsity 4} \label{fig12}
\end{figure}

\begin{figure}[h]
  \centering
  \centerline{\epsfig{figure=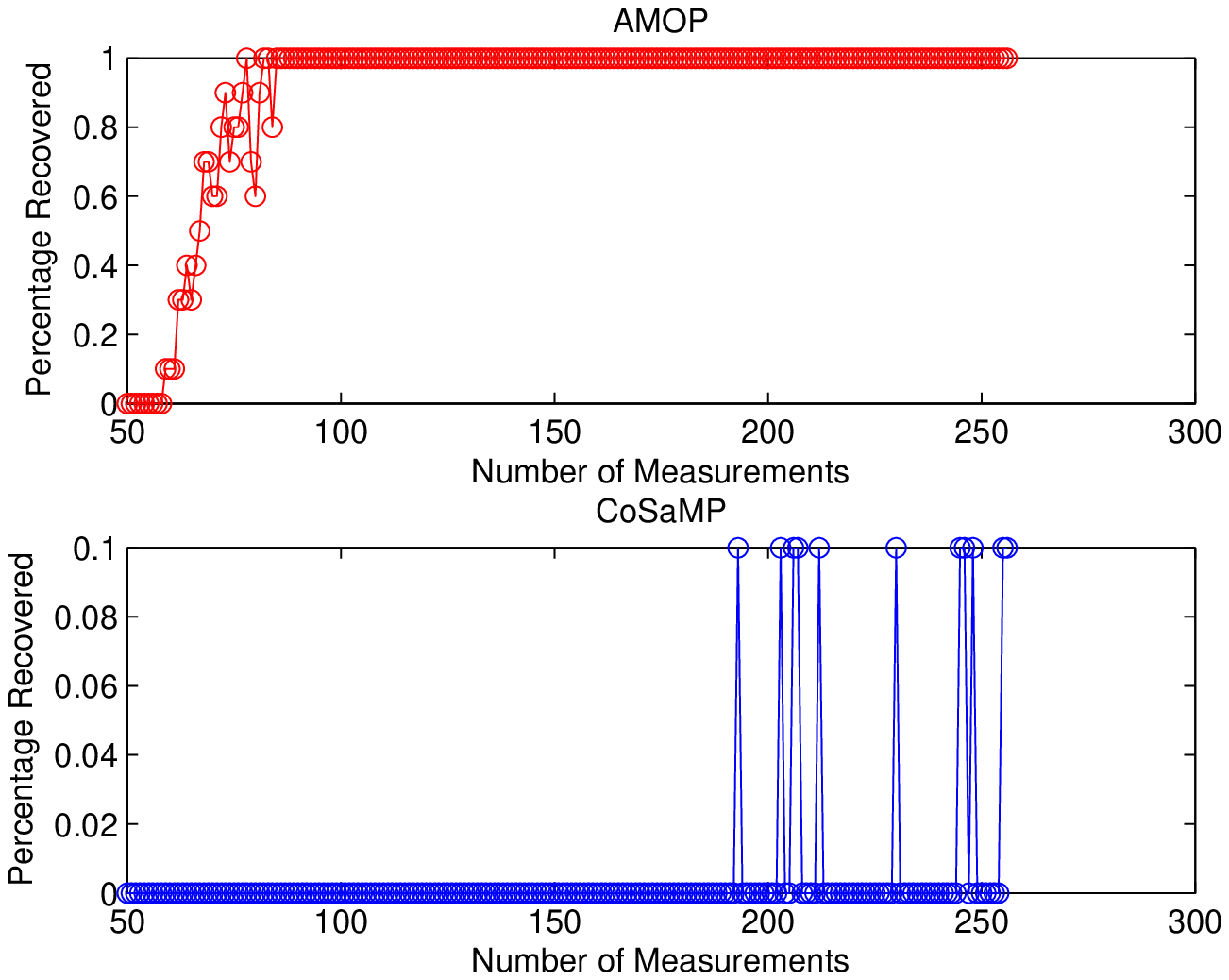,width=10cm}}
  \caption{Recovery Percentage of Piecewise Flat Signal with Sparsity 20} \label{fig13}
\end{figure}

\begin{figure}[h]
  \centering
  \centerline{\epsfig{figure=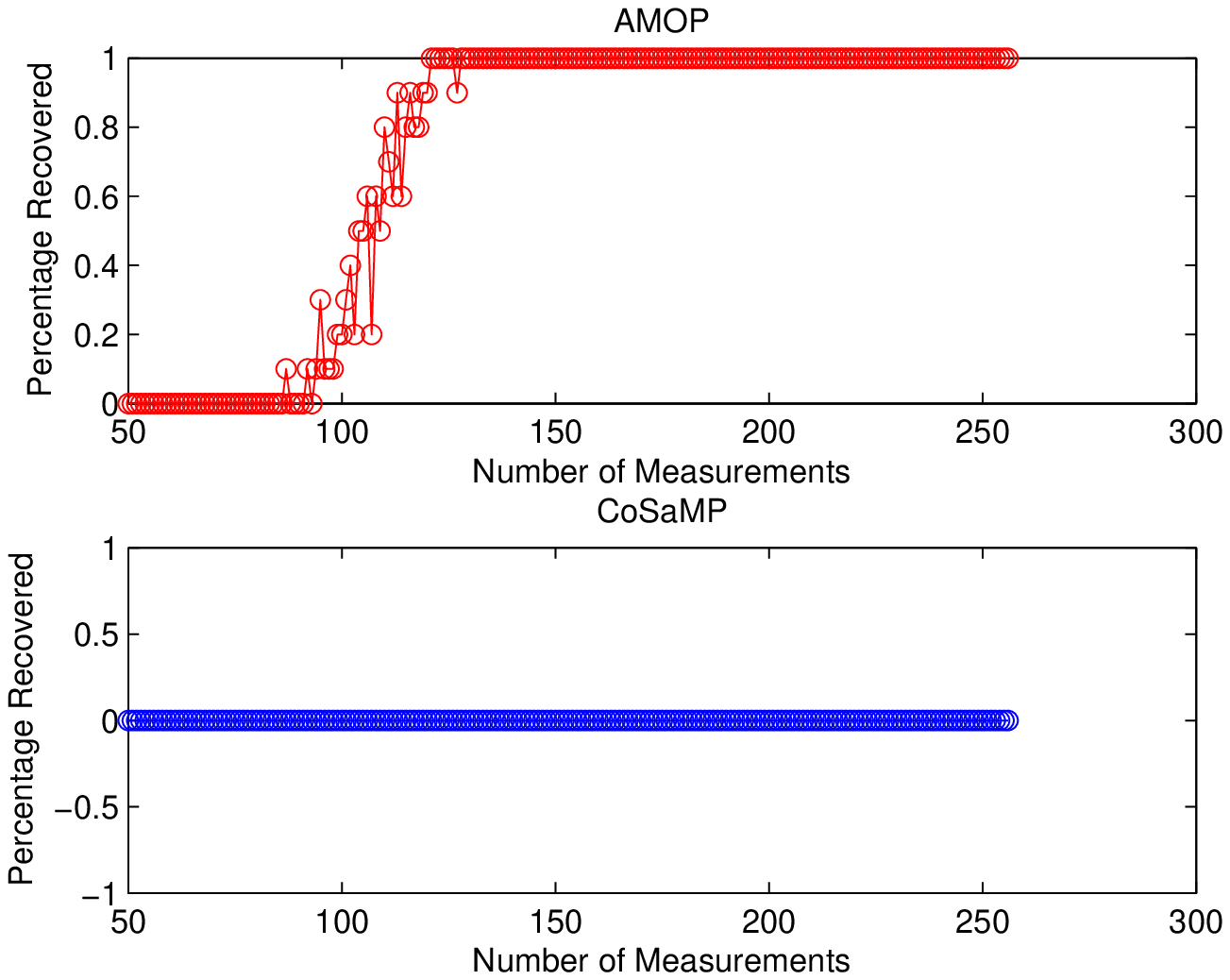,width=10cm}}
  \caption{Recovery Percentage of Piecewise Flat Signal with Sparsity 36} \label{fig14}
\end{figure}

Recovery percentage of CoSaMP in our experiment isn't ideal.
Especially when the number of non-zero elements of target signal is
large, the successful recovery probability of CoSaMP becomes more
and more ignorable. The capability of CoSaMP is doubtful in this
setting. On the contrary, the behavior of AMOP is very robust when
value distribution of target signal is changed. Furthermore, the
constant $C$ in \ref{label20} is approximately equal to that in flat
signal setting. Although it is predicted theoretically that constant
$C$ only depends on the construction of measurement matrix $\Phi$,
not rely on nature of target signal, it is common in practice that
the detailed value distribution of target signal certainly has
influence on performance of recovery algorithms. Our experimental
result indicates that the actual behavior of AMOP coincides with
theoretical conclusion perfectly. This argument is confirmed by
result for the other case of non-flat target signal. Here target
signal is chosen as compressible signal which is frequently
presented in orthogonal representation of signal, such as Discrete
Cosine Transform and wavelet transform, and data compression. Two
kinds of compressible signal are analyzed in our experiment, one is
exponential signal,
\[
x(n)= C*\alpha^n, \quad\quad{0<\alpha<1}
\]
the other is polynomial signal
\[
x(n)= C*n^{1/p}, \quad\quad{0<p<1}
\]
For brevity, only result for exponential signal is depicted in
\ref{fig15} and corresponding curve for CoSaMP is omitted. We also
performed the same experiment for polynomial compressible signals
and found the results very similar to those in Figure \ref{fig14}.

\begin{figure}[h]
  \centering
  \centerline{\epsfig{figure=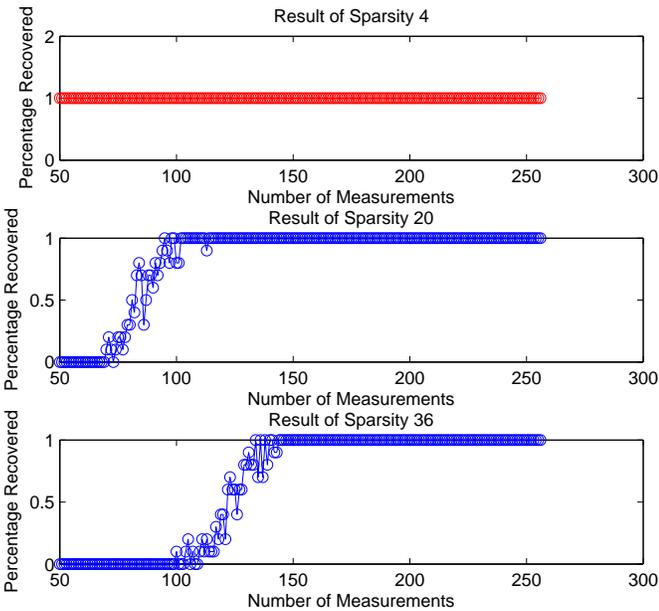,width=10cm}}
  \caption{Recovery Percentage of Polynomial Flat Signal} \label{fig15}
\end{figure}

\subsection{Target Signal With Noise}

Random noise was added in target signal to test the performance of
sparse recovery algorithms in noisy environment. Measurement matrix
is fixed to Fourier random matrix and target signal is set to flat.
The sparsity of target signal is fixed to 20 and size of
measurements $m$ is fixed to 200. The relative error of AMOP and
CoSaMP in Gaussian white noise with various level is depicted in
Figure \ref{fig16}

\begin{figure}[h]
  \centering
  \centerline{\epsfig{figure=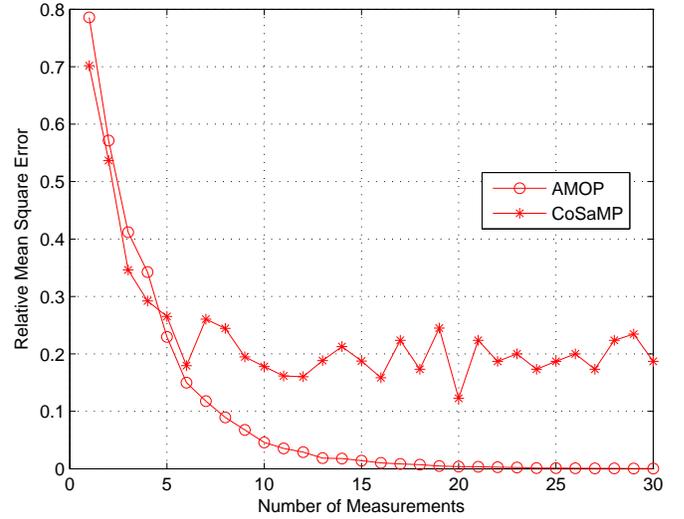,width=10cm}}
  \caption{Relative Recovery Error Under Gaussian Noise along with SNR} \label{fig16}
\end{figure}

The capacity of AMOP in noisy environment is satisfactory. Relative
recovery error could be controlled within $10\%$ when SNR is about
10dB. Even when SNR is as low as 5dB, relative error of AMOP still
could be governed within $20\%$. Though the error curve rise very
acutely in very low SNR region, it is shown that AMOP works normally
in most noisy environment. On the other hand, the relative error of
CoSaMP keeps on a high level when noise is presented. When SNR was
increased, it didn't exhibit the obvious trend of decreasing. With
the language of statistics, CoSaMP isn't consistent estimator.

\begin{figure}[h]
  \centering
  \centerline{\epsfig{figure=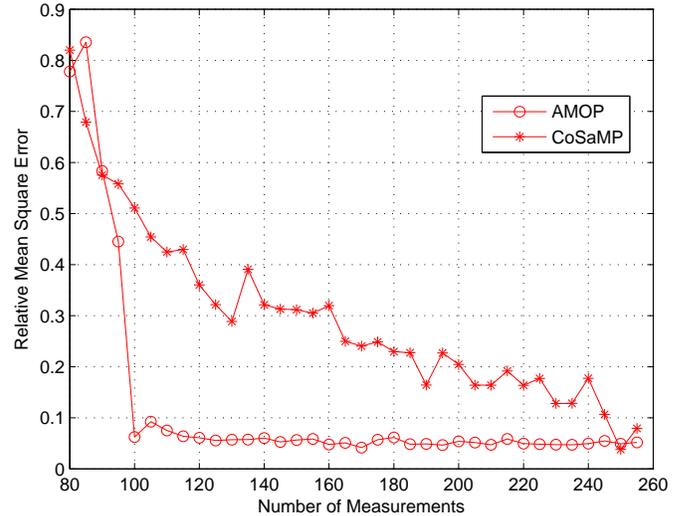,width=10cm}}
  \caption{Relative Recovery Error Under Gaussian Noise along with Measurement} \label{fig17}
\end{figure}

Figure \ref{fig17} illustrates the advantage of AMOP in noisy
environment from other viewpoint. Here SNR is fixed to 10dB, the
error with various size of measurements is plotted. It is observed
that if number of measurements $n$ is small, the performance of AMOP
is heavily abnormal. In fact, relative error of AMOP is even higher
than CoSaMP when $n$ is lower than 100. But it falls abruptly when
$m$ increases while that of CoSaMP remains in narrow range. When $m$
is larger than 100, the relative error of AMOP tends to be stable.
It is well controlled with in $10\%$, which is similar to Figure
\ref{fig16}.

\subsection{Measurement Matrix in STAP}

We would build a measurement matrix $\Phi$ with spatial-temporal
basis vectors, which is key component in theory and computation of
STAP (Space-Time Adaptive Processing), to investigate the potential
of sparse recovery algorithms to be applied in field of modern radar
and communication engineering. Here $\Phi$ is set to
\begin{equation}\label{label21}
\Phi=[\phi_{s-t}(1,1),\cdots,\phi_{s-t}(1,n),\cdots,\phi_{s-t}(m,n)],
\end{equation}
and
\begin{eqnarray}
\phi_{s-t}(f_s,f_d)&=&[1,\exp(j2{\pi}f_d),\cdots,\exp(j2{\pi}(L-1)f_d),\nonumber\\
&&\cdots,\exp(j2{\pi}((N-1)f_s+(L-1)f_d))]^T,\nonumber
\end{eqnarray}
where $f_d$ and $f_s$ denotes Doppler and spatial frequency,
respectively. Unlike Fourier matrix, the construction of
spatial-temporal matrix $\Phi$ is more complex. The phase of
elements in $\Phi$ composed of two parts, one is contributed by
Doppler frequency and the other by spatial frequency. It lead to
following consequence: Different $f_d$ and $f_s$ could be combined
to form the same (or approximately equal) phase. In other words,
strong correlation exists in different column vectors in $\Phi$,
which corresponding to various points far away with each other on
spatial-temporal plane. So Restricted Isometric Constant of $\Phi$
is conjectured to be relatively large. It is a challenge for sparse
recovery algorithm to be feasible when measurement matrix $\Phi$ is
chosen as spatial-temporal matrix.

Generally speaking, the support ("Position" of non-zero elements) of
target signal is much important than its detailed value. It is
indeed true in practical engineering discipline. For example, in
radar STAP processing, sparse recovery is utilized to estimate the
energy distribution of clutter and interference on spatial-temporal
plane from sample data directly. Because echo of clutter and
interference is much stronger than radar target, the detailed
amplitude and phase of clutter and interference isn't crucial. As
long as the accurate support ("Position") of clutter and
interference on spatial-temporal plane is found out, we can design
the efficient filter to suppress clutter and interference
effectively. Hence the most important feature of sparse recovery
algorithm applied to STAP processing is detecting the support of
target signal with great precision.

The experiment was performed to test the ability of AMOP and CoSaMP
to detect signal support. Measurement matrix $\Phi$ was set to
$224{\times}900$ spatial-temporal matrix as \ref{label21}. Target
signal is noise-free flat signal and its sparsity is fixed to 20.
Figure \ref{fig18} depicts the average result from 100 trials.

\begin{figure}[h]
  \centering
  \centerline{\epsfig{figure=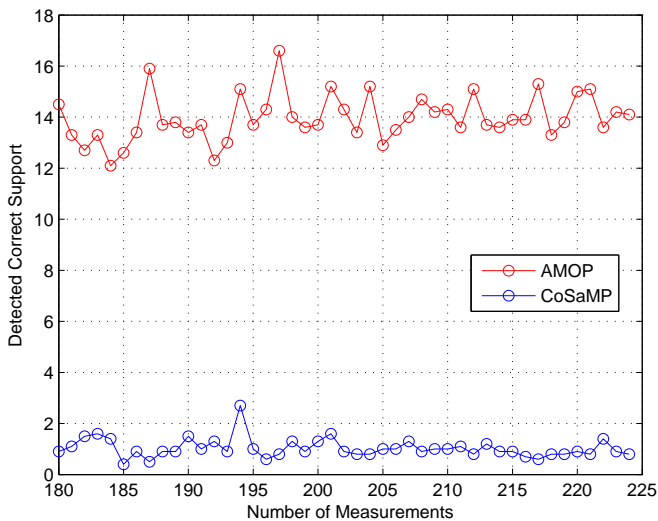,width=10cm}}
  \caption{Size of Detected Target Support With No Error} \label{fig18}
\end{figure}

It is clear that AMOP is much superior to CoSaMP when applied to
STAP processing because it could detect a majority of target support
with no error while CoSaMP could only find very few. But even so,
the accuracy of algorithms for support, whether AMOP or CoSaMP, can
hardly satisfy the requirement of practical STAP processors. Due to
ultra-low Signal Clutter Ratio (generally lower than -50dB), missing
four or five frequency points on spatial-temporal plane would lead
to very high false alarm rate and the performance of radar would
degrade heavily. So we should detect as much target support as
possible to minimize false alarm rate. If tiny error is allowed, the
behavior of sparse recovery algorithms becomes better, as depicted
in Figure \ref{fig19}

\begin{figure}[h]
  \centering
  \centerline{\epsfig{figure=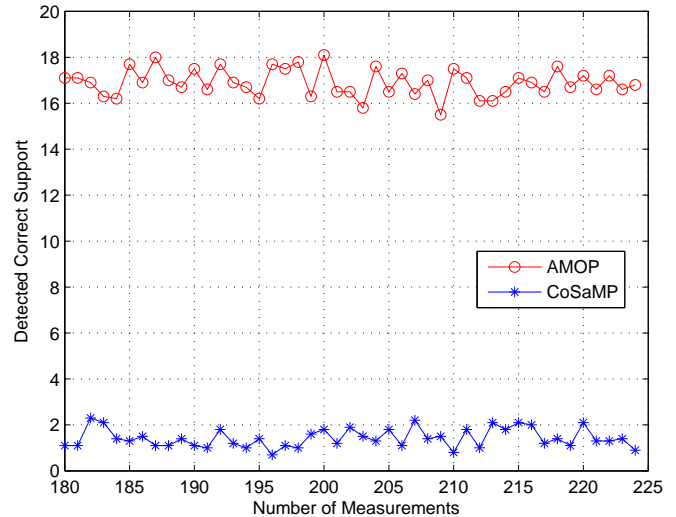,width=10cm}}
  \caption{Size of Detected Target Support With Error 1} \label{fig19}
\end{figure}

It is easily seen that if the "Position" detected having difference
1 with true "Position" of target signal is allowed to be counted,
the average size of detected target support for AMOP increases by
about 2 and that of CoSaMP is still negligible. It accounted for
that some support of target signal missed by AMOP wasn't really
missed. That is to say, their neighborhood, which corresponding to
the points adjacent to them on spatial-temporal plane, were
discovered instead. This is a good news for high performance filter
to suppress clutter could still be designed with AMOP and carefully
chosen notch, without losing much resolution and SNR. Figure
\ref{fig20} depicted the case of error 2. The behavior of AMOP
continued to be made better and approach the best. It seems that it
make little sense to enlarge error tolerance further.

\begin{figure}[h]
  \centering
  \centerline{\epsfig{figure=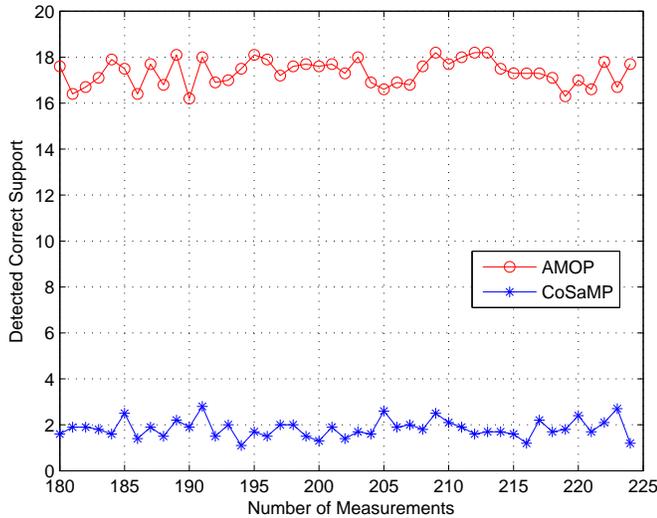,width=10cm}}
  \caption{Size of Detected Target Support With Error 2} \label{fig20}
\end{figure}

\section{Conclusion}

In this paper, a novel greedy algorithm for sparse recovery, called
AMOP, was given and examined. Its performance was studied by
theoretical analysis of simulation experiment.

The motivation of this algorithm is two obvious drawbacks in popular
methods, such as CoSaMP: Need of sparsity of target signal as prior
knowledge, and weak ability of working in noisy environment. With
well-designed algorithmic steps, AMOP can extract the information on
sparsity of target signal adaptively and sense the nature of target
signal automatically. It can recover the detail value of target
signal with very high precision with little prior knowledge. Its
validity is illustrated by strict deduction. Fine stability of
performance under random noise is another advantage of AMOP.
Furthermore, its robustness for various setting of target signal,
flat or compressible, and construction of measurement matrix, such
as spatial-temporal matrix, were also shown by thorough numerical
experiment. It is argued that AMOP is a excellent greedy algorithm
for sparse recovery and has great potential of widely utilization on
signal processing.

%
%

\end{document}